\documentclass[runningheads,a4paper]{llncs}

\usepackage{epstopdf}

\usepackage[american]{babel}
\usepackage{amssymb}

\usepackage{graphicx}
\usepackage{graphics}

\usepackage{subfigure}

\usepackage{paralist}

\usepackage{tikz}
\usetikzlibrary{decorations.pathreplacing}
\usetikzlibrary{calc}

\usepackage{color}
\usepackage{csquotes}

\usepackage[T1]{fontenc}

\usepackage{lmodern}
\usepackage{microtype}

\usepackage[math]{blindtext}

\ifnum\pdfoutput>0
\usepackage[
bookmarks=false,
breaklinks=true,
colorlinks=true,
linkcolor=black,
citecolor=black,
urlcolor=black,
pdfpagelayout=SinglePage
]{hyperref}
\usepackage[all]{hypcap}
\else
\usepackage{hyperref}
\fi

\usepackage{xspace}
\begin{document}

%\graphicspath{ {F:/Rakhshan/Courses/ConferencePaperCiE_10Jan/Images/} }
%Works on MiKTeX only
%hint by http://goemonx.blogspot.de/2012/01/pdflatex-ligaturen-und-copynpaste.html
%This allows a copy'n'paste of the text from the paper
\input glyphtounicode.tex
\pdfgentounicode=1

\title{A Nondeterministic Model for Abstract Geometrical Computation}
%If Title is too long, use \titlerunning
%\titlerunning{Short Title}

%Single insitute
\author{Rakhshan Harifi \and Sama Goliaei}
%If there are too many authors, use \authorrunning
%\authorrunning{First Author et al.}
\institute{ University of Tehran, Tehran, Iran \\ \email{\{rakhshan.harifi,sgoliaei\}@ut.ac.ir} }
%Multiple insitutes
%Currently disabled
%
\iffalse
%Multiple institutes are typeset as follows:
\author{Firstname Lastname\inst{1} \and Firstname Lastname\inst{2} }
%If there are too many authors, use \authorrunning
%\authorrunning{First Author et al.}

\institute{
Insitute 1\\
\email{...}\and
Insitute 2\\
\email{...}
}
\fi
			
\maketitle

\begin{abstract}
A signal machine is an abstract geometrical model for compu-
tation, proposed as an extension to the one-dimensional cellular automata,
in which discrete time and space of cellular automata is replaced with
continuous time and space in signal machine. A signal machine is defined
as a set of meta-signals and a set of rules. A signal machine starts from an
initial configuration which is a set of moving signals. Signals are moving
in space freely until a collision. Rules of signal machine specify what
happens after a collision, or in other words, specify out-coming signals
for each set of colliding signals. Originally signal machine is defined
by its rule as a deterministic machine. In this paper, we introduce the
concept of non-deterministic signal machine, which may contain more
than one defined rule for each set of colliding signals. We show that for
a specific class of nondeterministic signal machines, called $k$-restricted
nondeterministic signal machine, there is a deterministic signal machine
computing the same result as the nondeterministic one, on any given
initial configuration. $k$-restricted nondeterministic signal machine is a
nondeterministic signal machine which accepts an input iff produces a
special accepting signal, which have at most two nondeterministic rule
for each collision, and at most $k$ collisions before any acceptance.
\end{abstract}


%%%%%%%%%%%%%%%%%%%%%%%%%%%%%%%%%%%%%%%%%%%%%%%%%%%%%%%
\section{Introduction}
\label{sec:intro}
Consider some colored particles moving on a line with constant speeds. Some
particles may have zero speed and do not move. When two or more particles
collide, particles are replaced with some new colored particles, according to
predefined rules of the signal machine. Suppose that you are given some colored
particles with their initial position on a line and collision rules, and you are
asked to drawing the space-time diagram of the movement of these particles,
until no more collision may happen. Signal machines are such dynamical
systems, were signals are traces of particles in space-time diagram. The signal
machine model is an $\textit{Abstract Geometrical Computation}$ (AGC) model dealing
with Euclidean geometry which introduced in 2003 by Jerome-Durand Lose
for the first time \cite{durand2003calculer}. Signal machines are originally introduced as a continues extension of one-dimensional cellular automata in time and space \cite{durand2008signal}. Signals in
a signal machine are moving independently, thus, signal machines are inherently
parallel and they can viewed as a massively parallel computational model \cite{duchier2010massively}.

Signal machines are able to simulate any Turing-Computation and they are
Turing-Universal \cite{durand2005abstract}. With continuous time, they can be used to decide (in finite
time) recursively enumerable problems using the black-hole principle \cite{durand2005abstract,durand2006abstract}.
They are also capable of analog computation by using the continuity of space and
time to simulate analog models such as BSS \cite{durand2007abstract,durand2008abstract}. To achieve massive parallelism,
a fractal tree construction technique is provided on signal machine, and it is used
to solve the satisfiability of quantified Boolean formula (Q-SAT) problem, the
classic PSPACE-complete problem, in bounded space and time \cite{duchier2012computing}.

All studies on signal machines are on deterministic version of signal machines.
However, on every computational model, studying the non-deterministic versions
of computation and theirs computability power is an important concept. One
important question which arises is that if nondeterminism brings more power
to the considered computational model. To this end, the problem of comparing
computational power of nondeterministic and deterministic signal machines is
not investigated yet.

In this paper, first we define nondeterministic signal machine as a signal
machine which may have more than one rule applicable for each collision from
which one is nondeterministically selected to be applied. Then, we show that the
nondeterminism does not improve computability power of signal machines, for a
specific class of signal machines. For this purpose we use a constructive proof. In
other words, we propose an algorithm which converts each nondeterministic signal
machine to an equivalent deterministic signal machine. In our proposed algorithm
we utilize the parallel nature of signal machines and combine techniques and
structures which are applicable due to the geometrical nature of signal machines.
The main idea of our algorithm is to produce all possible paths of computation
in parallel. In this procedure, we try to obtain copies of the computations paths,
such that each possible path is assigned to one copy of the computation according
to a unique binary number.

In this paper, in Section \ref{sec:intro} we introduce signal machines and define it formally. In Section \ref {sec:struc} we describe many useful structures and abilities of signal machines which are useful for the rest of the paper. Afterward, in Section \ref{sec:NSM} first we introduce nondetrministic signal machines and its definition formally. 
Then, we propose a theorem on a class of nondeterministic signal machines and consequently, we propose an algorithm to prove it. 
Finally, in Section \ref{sec:conclusion} we propose some concluding remarks and possible future works.

%%%%%%%%%%%%%%%%%%%%%%%%%%%%%%%%%%%%%%%%%%%%%%%%%%%%%%%
\section{Background on Signal Machine}
\label{sec:intro}
The \textit{Signal Machine (SM)} model is an abstract geometrical computation model
which act based on many colored moving particles on an Euclidean linear space
and their collisions. Each signal machine has a set of collision rules. Collision rules
determine what happens after collision of a set of particles. If no rule is defined for
collision of a set of particles, nothing is replaced and the set of particles continue
their path. We can represent execution of a signal machine as a space-time
diagram of its particles. The trajectory of each particle in space-time diagram is
called a signal. Signal machines can be formally defined as $(M, S, R, c_0 )$, where:

\begin{itemize}
 \item $M$ is a finite, non-empty set of meta-signals. Each meta signal is a type
representing the set of signals with the same type. Thus, each signal is an
instance of a meta-signal $m \in M$.
\item $S:\ M \rightarrow {\rm I\!R}$ is a function which assigns a real speed to each meta-signal.
\item  $R : \{ \sigma_{1}, \sigma_{2},..., \sigma_{n} \} \rightarrow  \{\sigma_{1}',
\sigma_{2}',...,\sigma_{m}' \} $ is the collision function representing
signal replacement rules on collisions, where all $\sigma_{i}$ are meta-signals of distinct
speeds as well as $\sigma_{j}'$.
\item $c_0$ is the initial configuration.
\end{itemize}

An \textit{initial configuration} $c_{0}$ is a finite set $c_{0} = \{ ( \sigma_{i}, x_{i})| \sigma_{i} \in M\   and \  x_i \in {\rm I\!R} \}$
that determines the primary position of signal in space axis. A signal machine is executed from an initial configuration and 
represented geometrically as a space-time diagram. In the space-time diagram, time is increasing upwards. Figure \ref{fig:simple-signal}
illustrates a simple space-time diagram.

\begin{figure}[!htbp]
\begin{center}
\includegraphics[scale=1]{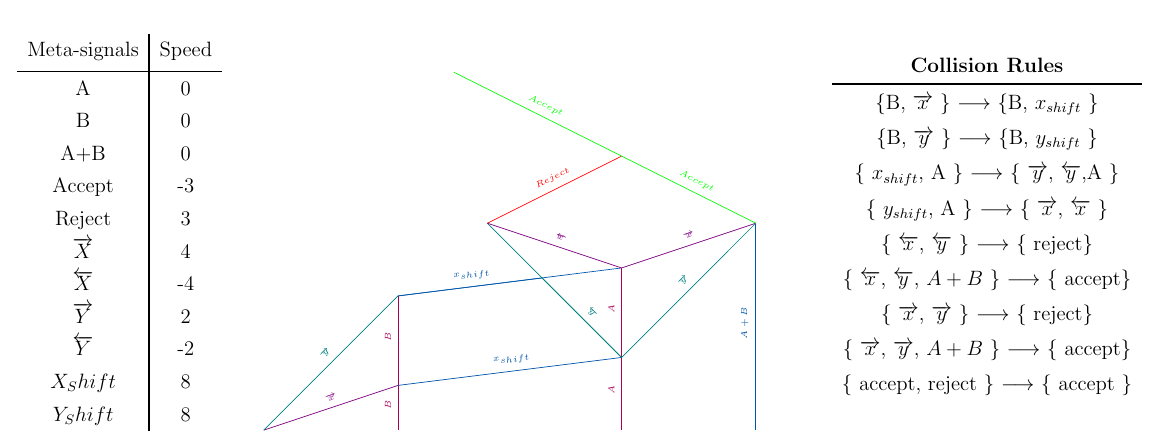}
\caption{Example of a simple signal machine. Meta-signals are given on the left and
collision rules are listed on the right side.}
\label{fig:simple-signal}
\end{center}
\end{figure}

The \textit{input} of a signal machine is defined by its initial configuration ($c_0$), and
the output is represented by signals that comes out when no more collision occurs.
Input of a signal machine is placed at the bottom of the time-space diagram at
time zero, and the output consists of signals moving freely after all the collisions
at the top of the time-space diagram.

We define a directed acyclic graph for a signal machine based on its collisions
and dependencies between collisions. This directed acyclic graph consists of
collisions as vertices and signals as edges, which are oriented according to their
moving direction on space-time diagram. Based on directed acyclic graph of
collisions, two major complexity measures are defined, 1) $\textit{time}$, which is defined
as the maximal length of a chain or collision depth, i.e. the length of the longest
path, and 2) $\textit{space}$, which is defined as the maximal number of signals in a
time \cite{duchier2010fractal}.

\section{Structures}%\paragraph{Freezing and Scaling}
\label{sec:struc}
In this section, we present many technique that could be applied on signal
machine structure due to the geometric nature of this model. These techniques
help us propose our algorithm.

\begin{proposition}(Middle) \cite{duchier2012computing}
Based on a simple geometric architecture, we can
obtain the middle point of two stationary signals, i.e. signals with zero speeds.
Figure \ref{fig:middle} illustrates this technique.
\end{proposition}

\begin{figure}[!htbp]
\begin{center}
\includegraphics[scale=1]{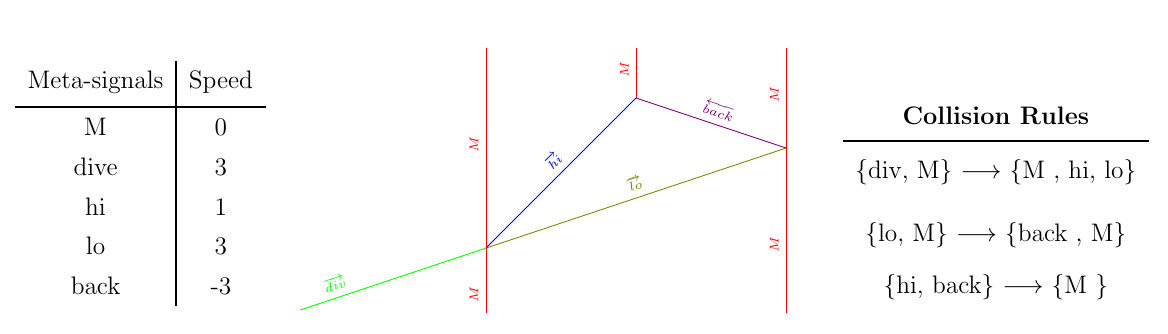}
\caption{Geometrical technique for obtaining the middle of two stationary signals (taken from \cite{duchier2012computing}).}
\label{fig:middle}
\end{center}
\end{figure}

\begin{proposition} (Freezing/Unfreezing) \cite{durand2009abstract}.
A signal machine is able to freeze
computations by replacing all the signals with a set of parallel signals during a
time and then unfreeze and continue the computation, after a while.
\end{proposition}
The $\textit{freezing}$ technique is used to preserve configuration and then restore it
later. During the freezing process, one freezing signal is sent from one end of the
space-time diagram and it crosses all the signals. In order to perform freezing
technique, for each meta-signal we define a new meta-signal representing frozen
version of it. A collision between the freezing signal and a signal replaces the
signal with the frozen version of it. Also, since the freezing signal may enter to
a collision of some non-freezing signals, for each set of meta-signals we define a
frozen version of that set, which is produced after a collision of freezing signal
with a set of non-frozen signals. All the frozen signals have same velocities, so they
produce parallel lines in space-time diagram and do not collide with each other. In addition, the distance between signals remain unchanged, so the configuration
freezes and shifts on space~\cite{adamatzky2012collision}.

To $\textit{unfreeze}$ a configuration, an unfreezing signal is sent from one end, crosses
the configuration, and replaces each frozen signal by the original one (or the
result of the collision, for frozen collisions). Freezing and unfreezing signals have
same speeds, thus, the configuration is restored exactly as it was before, but with
a shift in space.

Freezing is a useful technique and brings us a good ability in computations
by signal machines. For example, when a configuration is frozen, we can consider
parallel frozen signals as a set of signals, so it is possible to change their directions
and even propagate the beam to everywhere and then unfreeze and retrieve the
computations by an unfreezing signal with the same velocity as the freezing
signal. Freezing and unfreezing procedure is depicted in Figure \ref{fig:freezing}.

\begin{figure}[!htbp]
\begin{center}
\includegraphics[scale=0.8]{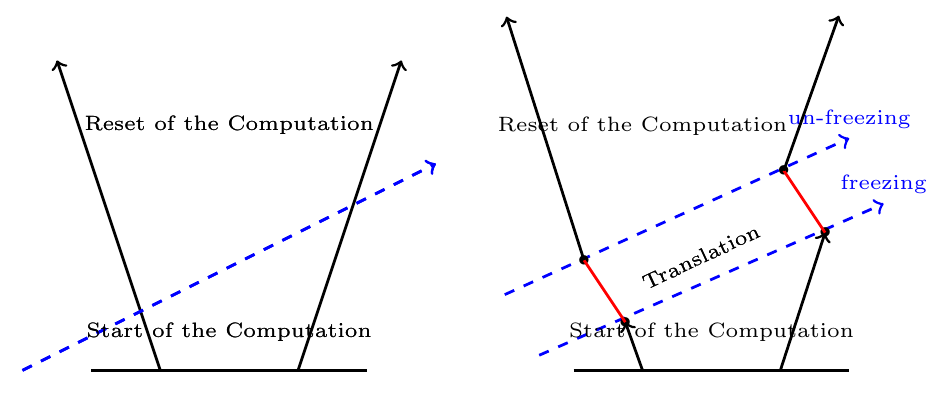}
\caption{Freezing technique (taken from \cite{durand2009abstract}).}
\label{fig:freezing}
\end{center}
\end{figure}

\begin{proposition} (Scaling). A signal machine can scale distances between parallel
signals, and thus each computation, by a scaling factor.
\end{proposition}

The scaling means to scale down the distances between signals without any
change in their velocities and rules. Thus, by scaling technique, the configuration
is scaled without any change in the computation of signal machine. Scaling of
parallel signals is easy by a Thales based construction. Thus, for scaling any
computation (not necessarily parallel signals), we can freeze the computations to a
set of parallel signals and then scale the parallel signals by a \textit{scaling} signal. Finally,
by a \textit{unfreezing} signal which is parallel to the \textit{freezing} signal, the computations
with smaller scale is unfrozen and continue. In Figure \ref{fig:scaling}, the scaling procedure
and an example of applying scaling procedure on a simple signal machine (of
Figure \ref{fig:simple-signal}) is presented.

\begin{figure}[!htbp]
\centering
\subfigure[Scaling parallel signals \cite{durand2009abstract}]{\includegraphics[scale=0.5]{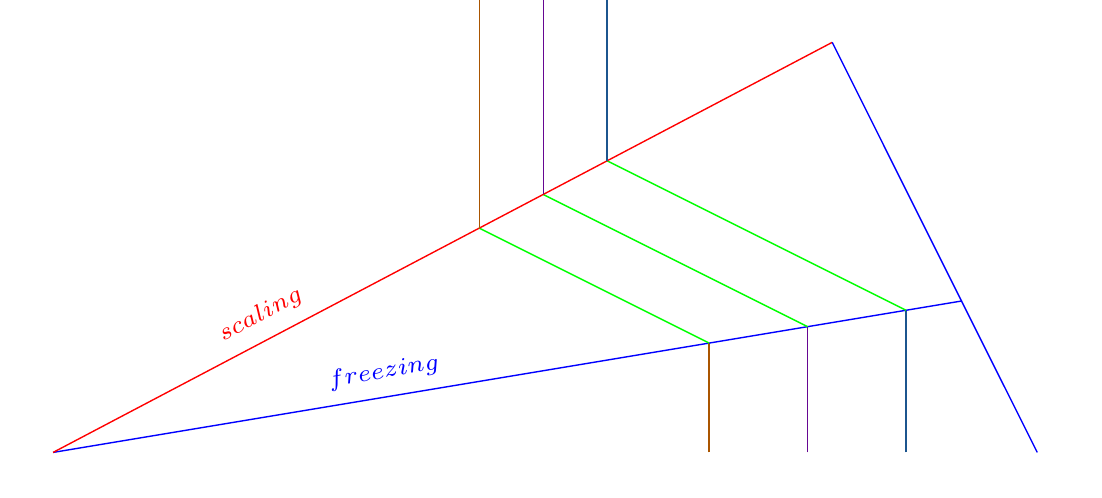}}
\subfigure[Scaling any computation \cite{durand2009abstract}]{\includegraphics[scale=0.5]{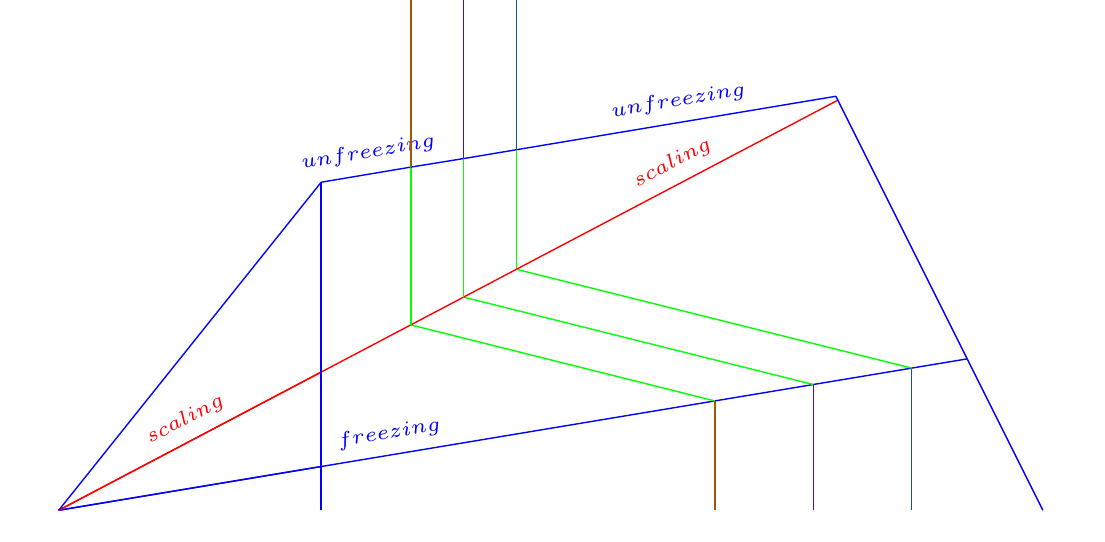}}
\subfigure[Example of scaling a simple signal machine]{\includegraphics[scale=1.0]{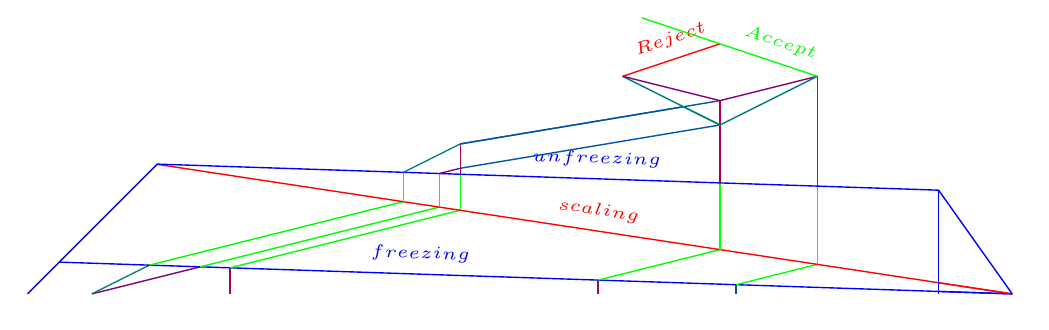}}
\label{fig:scaling}
\caption{Scaling.}
\end{figure}
 \begin{proposition}(Fractal Cloud)
Applying the procedure for obtaining the middle
of two signals repeatedly, a signal machine is able to halve space recursively and
have a fractal-like structure which is called fractal cloud (see Figure \ref{fig:fractal}).
\end{proposition}

This fractal cloud architecture is used for massively parallel computations. For
example, it is possible to propagate a computation's configuration as a beam
using the fractal cloud and do computations in each branch of this fractal cloud
in parallel.

\begin{figure}[!htbp]
\begin{center}
\includegraphics[scale=0.8]{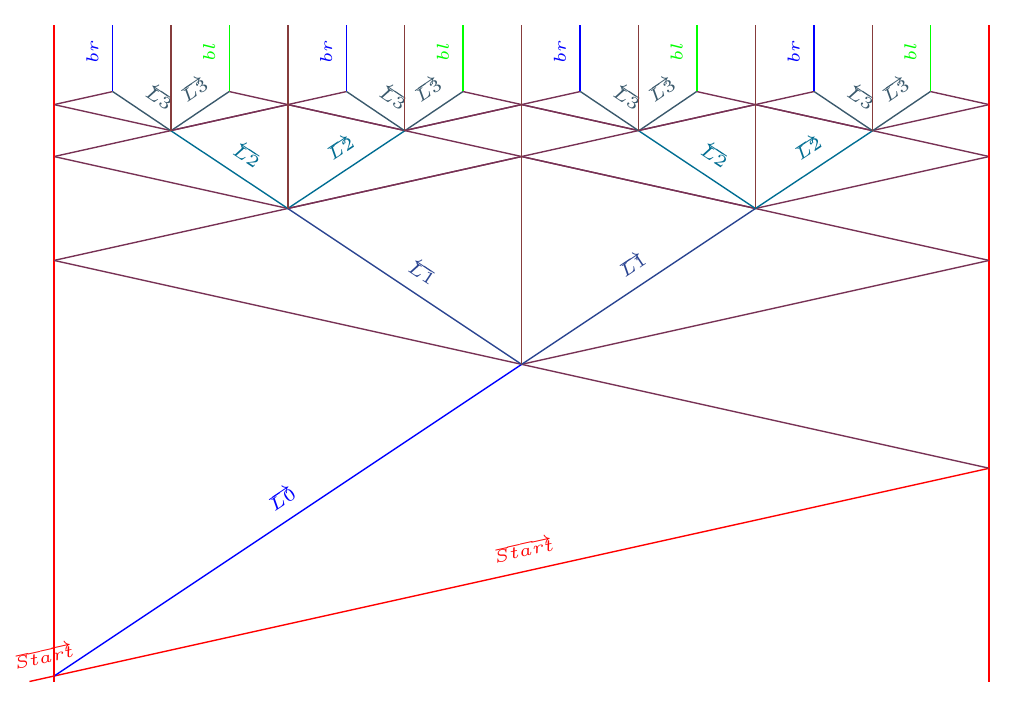}
\caption{Division process for constructing fractal cloud (taken from \cite{duchier2010fractal}).}
\label{fig:fractal}
\end{center}
\end{figure}

\section{Nondeterministic Signal Machine}
\label{sec:NSM}
In this section, we define nondeterministic signal machine, and show that deter-
ministic signal machine is capable of simulating a class of nondeterministic signal
machines. 

\subsection{Definition}
All the above mentioned signal machines are \textit{deterministic signal machines (DSM)}.
In this paper, we introduce \textit{nondeterministic signal machine (NSM)} as a signal
machine which may have more than one applicable rule for a collision. Figure \ref{fig:NSM}
presents a simple example of nondeterministic signal machine.
\\A \textit{nondeterministic signal machine}  is defined by a tuple $(M,S,R,c_0)$ where:
\begin{itemize}
\item $M$ is a finite set of \textit{meta-signals}.
\item $S : M \rightarrow {\rm I\!R}$ is a function which assigns a real speed to each meta-signal.
\item $R : \{ \sigma_{1}, \sigma_{2},..., \sigma_{n} \} \rightarrow  \{\sigma_{1}', \sigma_{2}',...,\sigma_{p}' \} | \{ \sigma_{1}'', \sigma_{2}'',..., \sigma_{q}'' \} |...$  is a set of collision rules which is a mapping from an arbitrary subset of $M$ with cardinality of at least two to any number of subsets of $M$, where all  $\sigma_{i}$ are meta-signals of distinct speed as well as  $\sigma_{j}'$ and $\sigma_{k}''$.
 \item $c_0$ is the initial configuration.
\end{itemize}
 We say that a nondeterministic signal machine accepts an input, if there
exists a set of collision rules when applying on collisions, an accepting output is
produced by the signal machine.
 \begin{figure}[!htbp]
\begin{center}
\subfigure[meta-signals]{\includegraphics[scale=0.5]{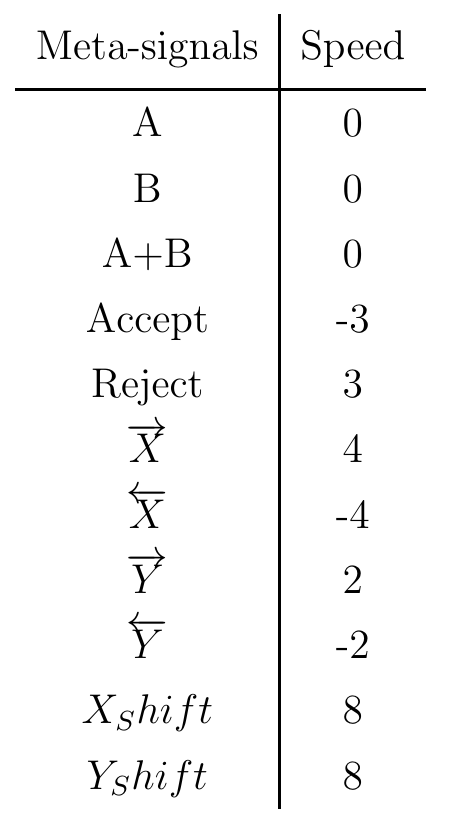}}
\subfigure[rules]{\includegraphics[scale=0.5]{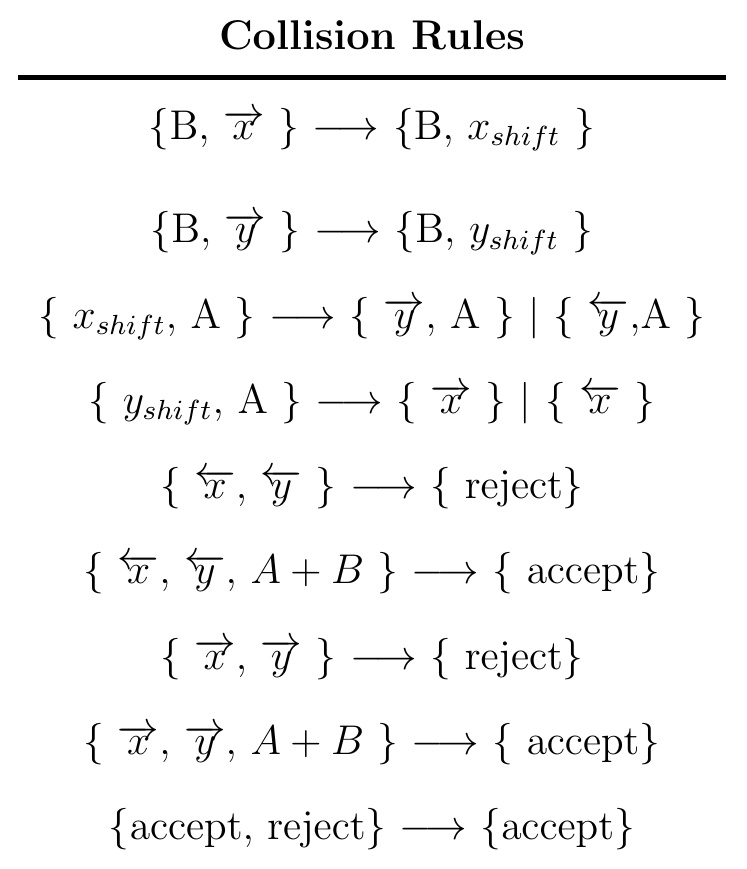}}
\subfigure[two examples of possible paths for this NSM ]{\includegraphics[scale=0.5]{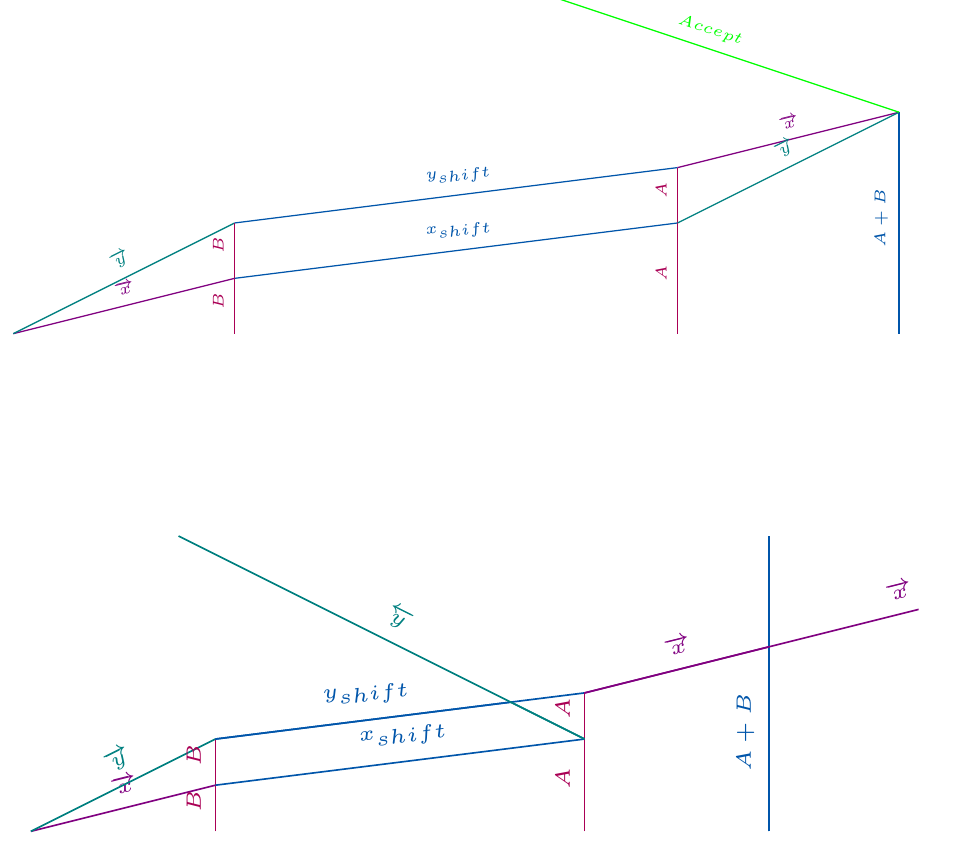}}
%\subfigure{\includegraphics[scale=0.5]{NSM2_cropped.pdf}}
%\subfigure{\includegraphics[scale=0.5]{NSM1(2)_cropped.pdf}}
%\subfigure{\includegraphics[scale=0.5]{NSM3(2)_cropped.pdf}}
\caption{A nondeterminestic signal machine with four possible diagram.}
\label{fig:NSM}
\end{center}
\end{figure}

 \subsection{Simulation of NSM by DSM}
 
 According to definition of NSM, every deterministic signal machine is a special
case of NSM. Now, we investigate if nondeterminism brings more
computational
power to signal machine or not.

We define \textit{k}-Restricted Nondeterministic Signal Machine (\textit{k}-RNSM) as a NSM
having following conditions:

\begin{itemize}
\item No two collisions occur in exactly same time.
\item At most two rules are defined for each collision.
\item An input is accepted by the machine if and only if a special signal accept is
produced during computation.
\item If an input is accepted by the machine, at most k collision occurs before
creation of the accept signal.
\end{itemize}
\begin{theorem} 
Let NN be an k-RNSM, then, there exists a deterministic signal
machine $D_{\textit{NN}}$ that accepts each input configuration if and only if NN accepts
the input.
\end{theorem}

\begin{proof}
The idea of the proof is to produce all possible paths of nondeterministic
computations of \textit{NN} in a fractal cloud, and test if one of them leads to
an accept signal.
\end{proof}

Since for each input, there are at most k collisions before production of
an accept signal in \textit{NN} , thus, we only have to check at most $2^k$ different
possible space-time diagrams which may be produced by \textit{NN} . We produce $2^k$
different paths by utilizing $2^k$ binary unique numbers and assign them to different
computation paths. Thus, computation in $D_\textit{NN}$ consists of two stages:
\begin{itemize}
  \item A fractal cloud with $2^k$ leaves (depth of \textit{k}) is produced, each leaf is used for
simulation of one possible space-time diagram of \textit{NN} . Thus, at the last level
of the fractal cloud, in each branch or leaf we have a unique binary number
and also an initial configuration of \textit{NN} . To represent $2^k$ binary numbers we
use \textit{k} stationary signals. We call a set of k signals, which represent a binary
number, a binary beam.
  
  \item At each leaf of fractal cloud, a possible space-time diagram of \textit{NN} for the
given input is simulated. In each possbile space-time diagram at most k
collisions are made. According to the binary beam of a branch, for each
collision, we choose between two possible collision rules and apply that.

\end{itemize} 
Note that if \textit{NN} accepts the input, the input signal will be produced in at least one leaf of fractal cloud in $D_\textit{NN}$.

\subsection{constructing a fractal cloud of $2^k$ binary numbers and initial configuration}
In order to construct $2^k$ binary beams, we combine the idea of decision tree and
architecture of fractal cloud with \textit{k} division levels \cite{duchier2010fractal} (see Figure \ref{fig:K2:A}). Also, in
each dividing level of constructing fractal cloud, we use an automatically scaling
procedure which is called \textit{lens device} \cite{duchier2012computing}. This is done in order to shrinking the
data and beam according to the structure in each division.

In order to construct the cloud, we start with a initial configuration of \textit{NN} and
initial configuration of a fractal cloud. First, we freeze and scale the whole initial
configuration of \textit{NN} coupled with a beam of $\{b_1, b_2, ... , b_k\}$ which is consists
of \textit{k} raw signals for constructing a binary number. 
Each signal $b_i$ is changed to either $0_i$ or $1_i$ at the $i$-th level. After freezing and scaling, we distribute frozen
data through the fractal cloud and make a decision for signal of $b_i$ at $i$-th level (see Figure \ref{fig:K2:B}).

\begin{figure}[!htbp]
\centering
\subfigure[]{\includegraphics[scale=0.5]{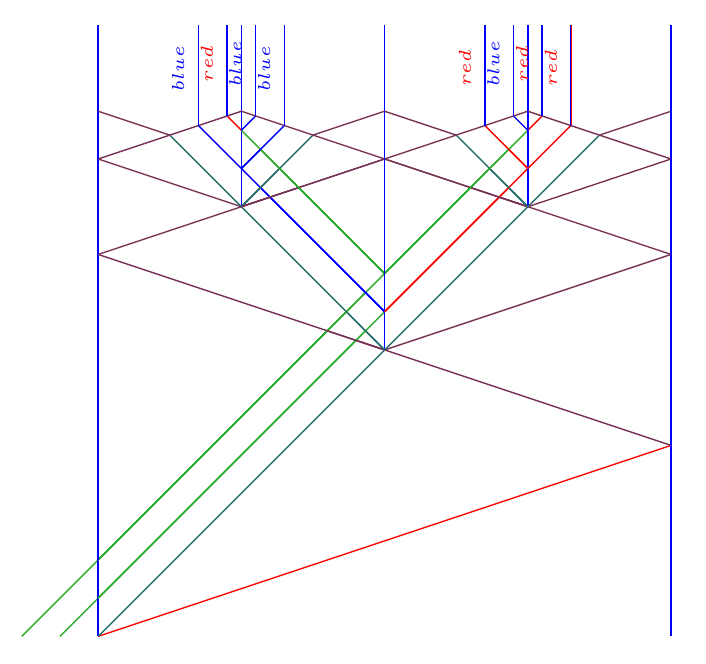}\label{fig:K2:A}}
\subfigure[]{\includegraphics[scale=0.75]{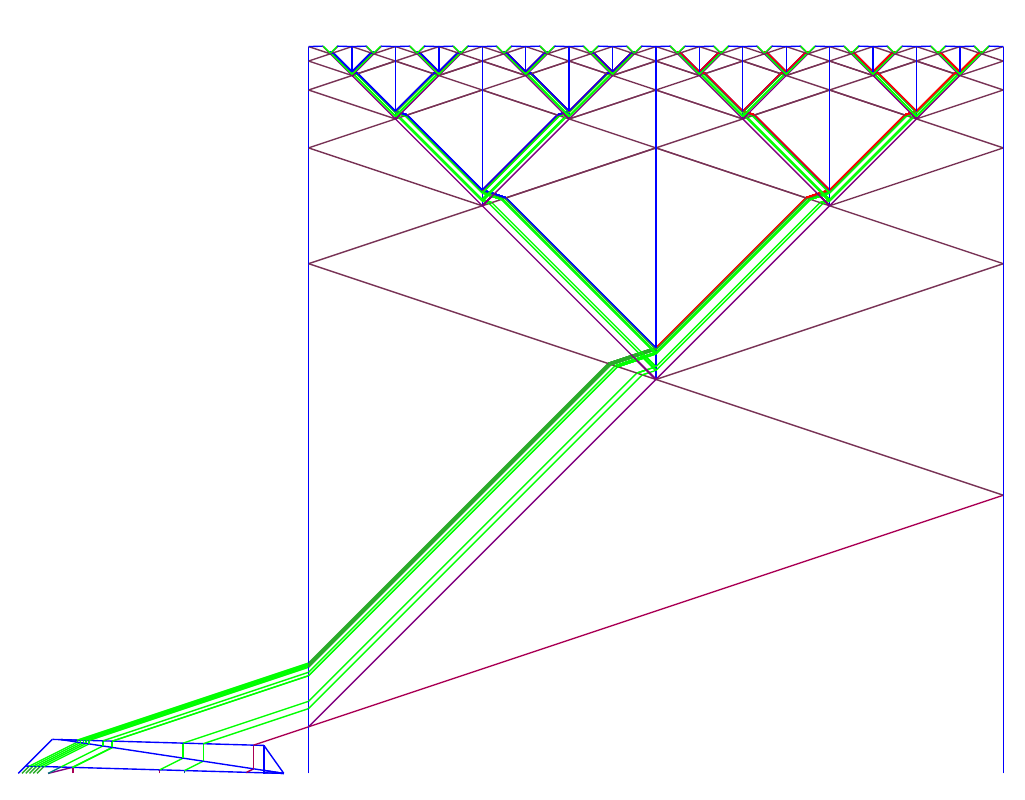}\label{fig:K2:B}}
\caption{binary decision tree and applying it on our complete fractal cloud.} 
\label{fig:K2}
\end{figure}

\subsection{ Executing one of NSM's paths deterministically according to the binary beam} 

Now suppose that we are at one leaf of our constructed fractal cloud and we
want to construct the corresponding space-time diagram. The idea is to simulate
computations of \textit{NN} collision by collision. For the $i$-th collision, we make a
decision on choosing one of the at most two possible collision rules according to
the value of ($b_i$), and apply it.

Thus, we start from the input configuration (a copy of input signals are
presented in each leaf of the tree), and at the $i$-the collision, we freeze the
computation and choose the collision. We suppose that the speed of freezing
signal is high enough that no other collisions occurs before freezing whole the
computations. A \textit{freezing} signal is send toward the binary beam, and then a
$message$ signal is send in opposite direction, which encodes the value of $b_i$. When
the $message$ signal meets the encoded collision signal, according to value of
message signal (\textit{True} or \textit{False}), one of the defined rules of \textit{NN} for the current
collision is applied. Also in this point, two unfreezing signals are send in both
direction with same speed as the freezing signal to unfreeze the computation.
Figure \ref{fig:collision point} represent how we approach in each collision point of \textit{NN}.

\begin{figure}[!htbp]
\begin{center}
\includegraphics[scale=1]{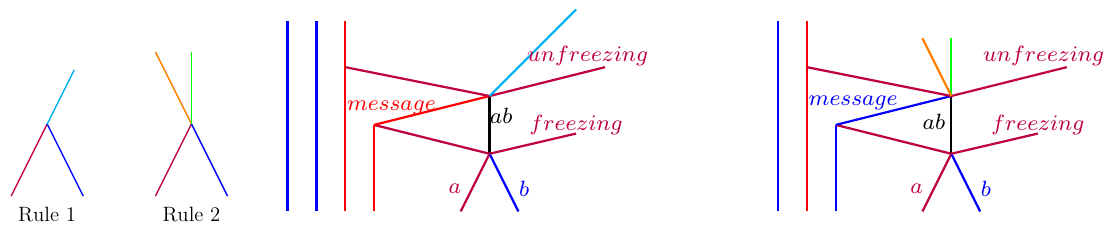}
\caption{Execution
at a collision point of a nondeterministic signal machine.}
 \label{fig:collision point}

\end{center}
%\subfigure{\includegraphics[scale=0.7]{Beam_F_cropped.pdf}}
%\subfigure{\includegraphics[scale=0.5]{NewCollisionRules_cropped.pdf}}
%\subfigure{\includegraphics[scale=0.5]{TwoRule.pdf}}
\end{figure}
Figure \ref{fig:NSM_To_DSM} shows a deterministic possible path obtained by a binary beam for
the NSM which is introduced in the Figure \ref{fig:NSM}.

Since there are at most $k$ collisions in \textit{NN} before producing the accept signal,
thus, if \textit{NN} accepts an input configuration, $D_\textit{NN}$ will produce an accept signal
in at least one of its leaves of the fractal cloud. Also, if \textit{NN} does not accept an
input, the accept signal will not be produced in any leaves of the combinatorial
comb in  $D_\textit{NN}$. Thus, $D_\textit{NN}$ produces the same output as \textit{NN} for any given input.

Note that the space complexity of $D_\textit{NN}$ for an input configuration is $O(k2^ks)$,
where $s$ is the space complexity of \textit{NN} over the input configuration.

\begin{figure}[!htbp]
\begin{center}
\label{fig:Example} 
\includegraphics[scale=1.0]{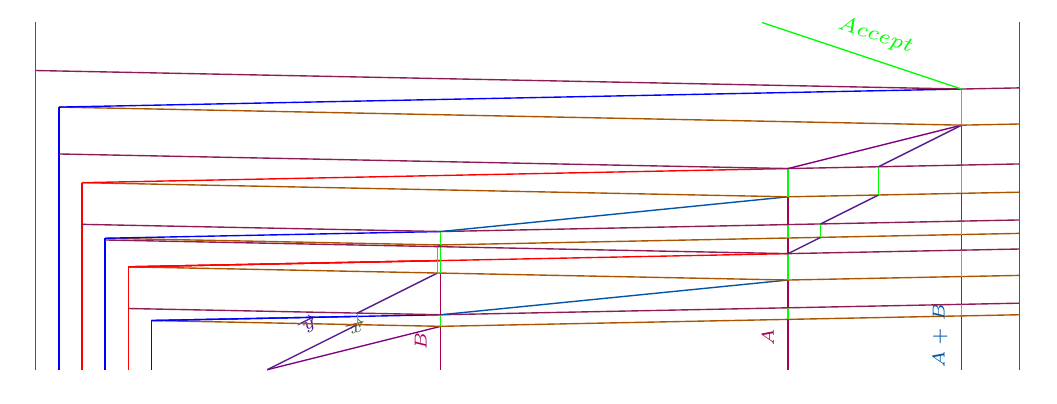}
\caption{Executing a NSM's path according to a binary number.}
\label{fig:NSM_To_DSM}
\end{center}
\end{figure}

\section{Conclusion and Future Works}
\label{sec:conclusion}

Signal machine is a model of computation inspired by information propagation
in cellular automata, where computation rules are defined according to signals
and their collisions. Although many computational aspects of signal machine is
previously investigated, but, the nondeterministic model of signal machines is still
remained unexplored. In this paper, we introduced the concept of nondeterministic
signal machine which is allowed to have more than one applicable rule for each
collision. We showed that for each nondeterministic signal machine under some
assumptions, there is a corresponding deterministic signal machine producing
same output on any given input configuration.

As future works, we will focus on more general classes of nondeterministic
signal machines. For example, we may consider the case that value $k$, maximum
number of collisions before acceptance, is unknown, and check that is there a
deterministic simulator for each (possibly unbounded) nondeterministic signal
machine. The idea we have is that we may guess the value $k$ and check simulate
a $\textit{k-RNSM}$. If it accepts after k collisions or has at most $k-1$ collisions, we
stop the simulation. Otherwise we may increase the guessed value k and repeat
the simulation. Also, we will try to check that whether the restriction of having
at most two nondeterministic rule for each collision reduces the computational
power of NSM or not.

%%%%%%%%%%%%%%%%%%%%%%%%%%%%%%%%%%%%%%%%%%%%%%%%%%%%%
\bibliographystyle{splncs03}
\bibliography{FinalReferences}

%%%%%%%%%%%%%%%%%%%%%%%%%%%%%%%%%%%%%%%%%%%%%%%%%%%%%

\end{document}